\documentclass[aps,prd,reprint,superscriptaddress,nofootinbib,longbibliography]{revtex4-2}
\usepackage{amsmath,amssymb,amsfonts}
\usepackage{orcidlink}
\usepackage{fnpct}
\usepackage{xcolor}
\usepackage{tikz}
\usepackage{hyperref}
\hypersetup{colorlinks=True}

\newtheorem{theorem}{Theorem}
\newtheorem{lemma}{Lemma}
\newtheorem{proposition}{Proposition}
\newtheorem{corollary}{Corollary}
\newtheorem{proof}{Proof}
\newtheorem{definition}{Definition}
\newtheorem{remark}{Remark}
\newtheorem{assumption}{Assumption}

\newcommand{\Dsc}{D_{\rm sc}}
\newcommand{\Msc}{M_{\rm sc}}
\newcommand{\Ramps}{\mathcal R_{\rm AMPS}}
\newcommand{\rev}[1]{\textcolor{blue}{#1}}
\newcommand{\revp}[1]{\textcolor{black}{#1}}
\begin{document}

\title{Entropy bound and the non-universality of entanglement islands}

\author{Naman Kumar\,\orcidlink{0000-0001-8593-1282}}
\affiliation{Department of Physics, Indian Institute of Technology Gandhinagar, Palaj, Gujarat, 382355, India}
\email{namankumar5954@gmail.com, naman.kumar@iitgn.ac.in}
\date{\today}
\begin{abstract}
The island prescription reproduces the Page curve and avoids the Almheiri--Marolf--Polchinski--Sully (AMPS) monogamy conflict for suitable radiation subsystems by placing the corresponding interior partners in their entanglement wedges. This mechanism is intrinsically region-dependent. We define a common compact semiclassical support as a finite semiclassical domain that contains representatives of the interior partners of a chosen family of late Hawking modes, while each radiation subsystem need reconstruct only its own partner. We ask whether these separate mode-by-mode reconstructions can be represented within such a common support at a fixed evaporation epoch. The physical motivation is that horizon normalcy---the requirement that the horizon remain smooth for an infalling observer and that semiclassical effective field theory remain valid across it---is a local property, whereas the island condition that resolves AMPS for a given late mode depends on the radiation subsystem used for reconstruction. Using a bounded null realization of the common support together with the quantum Bousso bound, we show that no regular common compact semiclassical support can realize all of the relevant AMPS partner sectors once the support enters the hyperentropic regime. The argument is formulated entirely within the semiclassical region in which the AMPS setup is valid and does not require extending the null surface to the Planckian singularity.
\end{abstract}

\maketitle

\section{Introduction}
\label{sec:intro}

The black hole information paradox arises from the tension between Hawking's semiclassical calculation, which gives an approximately thermal radiation state, and the expectation that black-hole formation and evaporation are described by unitary quantum evolution~\cite{Hawking:1975vcx}. The \revp{Almheiri--Marolf--Polchinski--Sully (AMPS)} argument sharpened this tension by adding the requirement of horizon smoothness~\cite{Almheiri:2012rt}. After the Page time~\cite{Page:1993wv,Almheiri:2020cfm}, a late Hawking mode $B$ must be purified by degrees of freedom in the early radiation in order to maintain unitarity. At the same time, semiclassical \revp{effective field theory (EFT)} near a smooth horizon requires $B$ to be purified by an interior partner, which we denote by $A_B$. If the radiation purifier and $A_B$ are independent systems, the two purifications violate monogamy of entanglement.

The island prescription provides a compelling semiclassical mechanism for avoiding this contradiction. For a radiation region $R$, the entropy is computed by extremizing the generalized entropy over candidate islands~\cite{Penington:2019npb,Almheiri:2019psf,Almheiri:2019hni,Almheiri:2019yqk,Almheiri:2019qdq,Penington:2019kki},
\begin{equation}
S(R)=\min_I\,\mathrm{ext}\left[
\frac{A(\partial I)}{4G}+S_{\rm out}(R\cup I)
\right],
\label{eq:island-formula}
\end{equation}
\revp{where $G$ is Newton's constant and $S_{\rm out}(R\cup I)$ is the renormalized bulk-matter entropy of the quantum fields on $R\cup I$.} After the Page time, the dominant saddle can include an island $I(R)$, so that interior degrees of freedom lie in the entanglement wedge of $R$. In that case, the interior partner and the radiation purifier need not be independent systems; they can be different representations of the same logical degree of freedom.

This paper does not question the standard, region-dependent island prescription. \revp{We ask whether the family of subsystem-dependent AMPS reconstructions can be represented within one compact semiclassical support common to all relevant late modes. If the answer is negative, then region dependence is not merely a convenient way of implementing the island prescription. Rather, it becomes necessary to avoid the common-support obstruction in which a single finite semiclassical support becomes hyperentropic and incompatible with the quantum Bousso bound \cite{Bousso:1999xy,Strominger:2003br,Bousso:2014sda,Bousso:2015mna}.} The motivation is that horizon normalcy---the statement that the horizon is smooth for an infalling observer and that semiclassical gravity and EFT remain
valid across it---is ordinarily a local property of the semiclassical geometry \cite{Bousso:2025udh}, whereas the island condition resolving AMPS for a late mode is subsystem-dependent. If $R_B$ is a radiation subsystem from which the purifier of $B$ can be decoded, the relevant condition is
\begin{equation}
A_B\in {\rm EW}(R_B),
\label{eq:local-condition-intro}
\end{equation}
\revp{where ${\rm EW}(R_B)$ denotes the entanglement wedge of $R_B$.} \revp{For a family $\{(B_i,R_i)\}$, we require only that the partner of $B_i$ have a representative\footnote{Here a representative means a realization of the corresponding logical partner within the stated semiclassical subdomain.}
in a subdomain $d_i$ of one common semiclassical domain $\mathcal D_*$ and that $d_i\subset {\rm EW}(R_i)$. Thus each $R_i$ reconstructs its own partner, not the rest of $\mathcal D_*$. This is the common-support condition tested below.}

The scope of the argument is crucial. AMPS is formulated in semiclassical effective field theory in a regular near-horizon region. We therefore introduce a cutoff semiclassical spacetime $\Msc$ and restrict every domain-of-dependence and null-surface statement to $\Msc$. A bounded null realization in this paper is not a global null completion of the entire black-hole interior and is not required to reach, cross, or resolve the Planckian singularity. It is a bounded representative of the same AMPS support domain inside $\Msc$. \revp{We formulate common support at a fixed evaporation time $u$. The region $B_*(u)$ may move as the black hole evolves; ``common'' means only that the same support is used for the chosen family of decoders at that epoch.}

The logic of the result has two steps. First, the partners of operationally distinguishable late Hawking modes, meaning modes that can be resolved as distinct semiclassical excitations, load independent logical sectors into a \revp{common} support, while the boundary area of that support remains bounded by the black-hole area. The support therefore becomes hyperentropic. \revp{Second, if a bounded quantum lightsheet has the same domain of dependence as the support and obeys the quantum Bousso bound, it requires the opposite inequality. Hence, after the crossover, either the common support fails, the bounded semiclassical null realization fails, or the quantum Bousso bound fails to apply.} If bounded null realizability and the entropy bound are imposed as requirements for a regular semiclassical common-support realization of the AMPS interior, then no \revp{common compact partner support} satisfying these conditions can exist.

\revp{Moreover, the entropy comparison follows the same domain-of-dependence logic used in entropy-based singularity theorems. Bousso and Shahbazi-Moghaddam use $D(B)=D(L)$ together with unitarity to identify the entropy of spatial and null representatives~\cite{Bousso:2022cun}. For the entropy inequality itself, however, we use a quantum Bousso bound adapted to semiclassical gravity rather than the classical covariant bound. The relevant null segment is seeded by a quantum extremal surface (QES) and controlled by restricted quantum focusing \cite{Shahbazi-Moghaddam:2022hbw}.}

This formulation also clarifies the relation to entropy-based singularity arguments. Such arguments diagnose null incompleteness or obstruction to further semiclassical evolution when a hyperentropic region cannot possess the required regular null development~\cite{Bousso:2022cun,Bousso:2022tdb}. Their conclusion is not that a null surface must be followed into a pre-existing Planckian singularity. Rather, the failure of a regular bounded semiclassical development is itself the obstruction. In the present problem this failure is one branch of the \revp{common-support} dichotomy.

The paper is organized so that the physical question precedes the technical theorem. Section~\ref{sec:normalcy} states the mode-by-mode AMPS criterion and the global-local gap. Section~\ref{sec:definitions} defines the semiclassical domain and all terms entering the theorem. Section~\ref{sec:qes} constructs the relevant quantum-lightsheet candidate. Section~\ref{sec:obstruction} states the entropy assumptions, the hyperentropic crossover, and the \revp{common-support} dichotomy. Section~\ref{sec:discussion} explains the implications for region dependence, time dependence, covariance, and solvable island models.

\section{Horizon normalcy, the global-local gap and covariance}
\label{sec:normalcy}

We now develop the physical motivation for the common-support question before turning to the technical definitions and entropy theorem. The island prescription is often invoked not merely to reproduce the Page curve, but also to explain why the horizon remains smooth for an infalling observer. These are related but logically distinct claims. Page-curve restoration is a global entropy statement for a selected radiation system, whereas horizon normalcy is a local, mode-by-mode statement about the semiclassical neighborhood of a horizon-crossing event \cite{Bousso:2025udh}. The common-support question asks whether these region-dependent reconstructions can be represented within one compact semiclassical interior support.

\subsection{AMPS after the Page time}
\label{subsec:amps}

Let $B$ be a late Hawking wave-packet mode near the horizon, localized in the regular semiclassical regime in which the AMPS argument is formulated, and let $R$ denote the radiation already emitted by the time of the horizon crossing. After the Page time, unitarity implies that $B$ is purified by degrees of freedom in the radiation. More precisely, there exists a decoded mode $\widetilde B\subset R_B\subseteq R$ such that
\begin{equation}
S(B\widetilde B)\ll 1.
\label{eq:decoded-purifier}
\end{equation}
Horizon normalcy requires a different-looking purification: $B$ must be entangled with an interior partner $A_B$ localized behind the horizon,
\begin{equation}
S(BA_B)\ll 1,\qquad S(B)=O(1).
\label{eq:interior-purifier}
\end{equation}
If $A_B$ and $\widetilde B$ are independent systems, Eqs.~\eqref{eq:decoded-purifier} and~\eqref{eq:interior-purifier} violate monogamy. The island prescription resolves this by making $A_B$ part of the radiation entanglement wedge. Then $A_B$ and $\widetilde B$ are not independent degrees of freedom; they are two descriptions of the same logical mode.

The condition needed for the mode $B$ is therefore
\begin{equation}
A_B\in {\rm EW}(R_B),
\label{eq:AB-in-EW}
\end{equation}
where $R_B$ is a radiation subsystem from which the purifier of $B$ can be decoded. This is the precise local condition behind the island resolution of AMPS. It is imposed in the regular semiclassical domain containing the late mode and its EFT partner; it is not a statement about the Planckian completion of the black-hole interior.

\subsection{The global-local gap}
\label{subsec:global-local}

The Page curve is a global entropy statement. In the standard calculation, one chooses a radiation system $R$, evaluates Eq.~\eqref{eq:island-formula}, and finds a post-Page saddle with a nontrivial island. This shows that the entropy of that chosen system follows the expected unitary behavior. It does not automatically show that every late mode $B$ has its particular partner $A_B$ in the wedge of every operationally relevant decoding subsystem $R_B$.

This is the global-local gap:
\begin{equation}
\begin{split}
&\text{Page-curve restoration} \\
&\qquad \not\Rightarrow\quad
A_B\in {\rm EW}(R_B)\;\text{for every late mode }B .
\end{split}
\label{eq:global-local-gap}
\end{equation}
The distinction matters because AMPS is formulated mode by mode. The paradox is avoided for a mode $B$ precisely when Eq.~\eqref{eq:AB-in-EW} holds. Thus the island prescription gives a conditional local resolution:
\begin{equation}
\text{AMPS avoided for }B
\quad\Longleftrightarrow\quad
A_B\in {\rm EW}(R_B).
\label{eq:amps-iff}
\end{equation}
This conditionality is not a failure of islands; it is a feature of entanglement wedge reconstruction. But it shows why Page-curve restoration alone does not exhaust the content of horizon normalcy.

A common support would close the gap in the minimal way needed for a single common semiclassical interior. What we mean by this is that every partner is represented in one compact semiclassical domain, while each radiation subsystem reconstructs only its own partner. %The stronger requirement that the entire domain lie in every entanglement wedge is not imposed.
The theorem below tests whether this weaker common-support completion can remain a regular semiclassical object. The bounded-null condition introduced later concerns only this finite support inside the semiclassical domain, not a global extension through the Planckian singularity.

\subsection{Region dependence and horizon normalcy}
\label{subsec:region-normalcy}

The region dependence of islands is not by itself problematic. Entanglement wedge reconstruction is naturally subsystem-dependent, and different boundary or radiation regions can reconstruct different bulk regions. The issue arises because the property being reconstructed is horizon normalcy. Smoothness of the horizon is normally regarded as a local feature of the semiclassical spacetime, not as a property that changes with the auxiliary radiation subsystem chosen by an outside decoder.

In this sense, the island resolution of AMPS has two logically distinct layers. The first is the Page-curve layer: the entropy of a chosen radiation system is computed correctly by including islands. The second is the normalcy layer: the interior partner required by local effective field theory must be identified with radiation degrees of freedom so that monogamy is not violated. The first layer is global and entropic; the second is local and mode-dependent. The \revp{common-support} question asks whether the second layer can be promoted from a family of region-dependent statements to a single compact-support statement.

The entropy theorem answers this question negatively under the stated semiclassical assumptions. If all AMPS-relevant partners are represented in one compact support, \revp{entropy eventually exceeds the amount allowed by the boundary area.} Thus \revp{even the weaker common-support realization of normalcy} is not available as a regular semiclassical construction after the hyperentropic crossover. The ordinary region-dependent prescription avoids the contradiction precisely by not requiring such a support.

This conclusion does not say that the individual region-dependent conditions $A_B\in {\rm EW}(R_B)$ are inconsistent. Rather, it says that their validity for appropriate radiation systems cannot in general be replaced by a single compact support common to every decoding problem. Region dependence is therefore not merely a technical feature of how the island formula is evaluated; it is part of how the framework avoids overloading one finite semiclassical interior domain with all late-time logical partners\footnote{By a logical partner we mean the interior degree of freedom that is entangled with a given late Hawking mode in the semiclassical code subspace. The term ``logical'' emphasizes that this degree of freedom is defined at the level of the encoded semiclassical information and need not be identified with the algebra of a sharply localized spacetime region.}
.

\subsection{Covariance tension}
\label{subsec:covariance}

The normalcy issue becomes sharper when one considers covariance. A semiclassical spacetime admits many valid Cauchy foliations. Different slices can intersect the outgoing Hawking radiation in different spatial regions, and the radiation subsystem used to decode a given late mode can therefore be represented as $R_B(\Sigma)$ on a slice $\Sigma$. The island criterion becomes
\begin{equation}
A_B\in {\rm EW}\bigl(R_B(\Sigma)\bigr).
\label{eq:slice-condition}
\end{equation}
Since the wedge is obtained by extremizing the generalized entropy for the chosen region, Eq.~\eqref{eq:slice-condition} is not manifestly a slice-independent statement about the same horizon-crossing event.

This does not unambiguously establish a violation of general covariance. Rather, we call this a \emph{covariance tension}. A local geometric statement, namely horizon normalcy, is being represented by a reconstruction condition that depends on a choice of radiation subsystem and potentially on how that subsystem is realized on a Cauchy slice. In ordinary semiclassical gravity, different foliations should describe the same local horizon physics. If no-drama is encoded only through Eq.~\eqref{eq:slice-condition}, then the island prescription by itself does not give a manifestly covariant normalcy criterion.

This tension is complementary to Bousso's argument that horizon normalcy cannot consistently depend on the state of distant Hawking radiation without either violating general covariance or introducing extreme nonlocality~\cite{Bousso:2025udh}. In that argument, the central concern is state dependence: whether a horizon-crossing event is normal can depend on which radiation state is present on a distant slice. Here the concern is reconstruction dependence: even before specifying a state-dependent interior map, the condition $A_B\in {\rm EW}(R_B)$ inherits the subsystem dependence of entanglement wedge reconstruction. The two concerns point in the same direction: a fully satisfactory account of the horizon interior should explain how local normalcy emerges as a covariant property rather than as a slice- or subsystem-dependent reconstruction statement.

The \revp{common-support} question is therefore a diagnostic of this covariance tension. A single compact support would provide one geometric interior for all mode-by-mode reconstructions without requiring every decoder to reconstruct the whole interior. The result below shows that this candidate cannot simultaneously \revp{contain all partner sectors, admit the bounded semiclassical null realization expected of a regular AMPS support,} and remain compatible with the \revp{assumed semiclassical entropy bound}. The obstruction leaves the standard family of region-dependent reconstructions intact, while showing why it cannot be \revp{represented by one compact, slice-independent semiclassical interior support}.

\section{Semiclassical scope and precise definitions}
\label{sec:definitions}
We now give precise definitions of the terms that are crucial for the discussion of our result.

\begin{definition}[Semiclassical AMPS domain]
\label{def:semiclassical-domain}
Let $\Msc$ be a globally hyperbolic region of the evaporating spacetime in which curvatures, stress-tensor fluctuations, and backreaction remain within the domain of semiclassical effective field theory. The region contains the near-horizon wave packets $B$, their EFT partners $A_B$, the relevant quantum extremal surfaces, and the portions of the radiation systems used in the AMPS experiment. Its boundary may include regular cutoff surfaces chosen before Planckian curvature or other breakdown of the semiclassical approximation. For an achronal set $X\subset\Msc$, define
\begin{equation}
\Dsc(X)\equiv D_{\Msc}(X),
\label{eq:Dsc-definition}
\end{equation}
the domain of dependence computed within $\Msc$.
\end{definition}

All causal statements below are statements in $\Msc$. In particular, no assumption requires a null generator to extend to the global black-hole singularity. It is important to note that the cutoff is not used to remove an entropy obstruction; it specifies the regime in which AMPS, the QES prescription, and the semiclassical entropy bound are simultaneously meaningful.

\begin{definition}[AMPS-relevant radiation region]
\label{def:amps-region}
For an operationally distinguishable late mode $B$, a radiation subsystem $R_B$ is AMPS-relevant if: (i) its algebra contains a decoded purifier $\widetilde B$ satisfying Eq.~\eqref{eq:decoded-purifier}; (ii) the generalized-entropy extremization defining ${\rm EW}(R_B)$ is well-defined in the semiclassical code subspace under consideration; and \revp{(iii) the QES and the code-subspace reconstruction map needed to test $A_B\in {\rm EW}(R_B)$ are both contained within $\Msc$.} The collection of all such subsystems is denoted by $\Ramps$.
\end{definition}

The definition does not require Eq.~\eqref{eq:AB-in-EW} to hold. It identifies precisely the subsystems for which that condition must be tested in order to resolve AMPS for the corresponding mode.

\begin{definition}[\revp{Common compact partner support}]
\label{def:common-support}
\revp{Fix an evaporation time $u$ and a family $\{(B_i,R_i)\}_{i=1}^{N(u)}$. Let $\Sigma_u$ be a Cauchy slice of $\Msc$ and let $B_*(u)\subset\Sigma_u\cap\Msc$ be a compact partial Cauchy region with edge $\sigma_*(u)\equiv\partial B_*(u)$. Define}
\begin{equation}
\mathcal D_*(u)\equiv\Dsc\bigl(B_*(u)\bigr).
\label{eq:support-domain}
\end{equation}
\revp{The region $B_*(u)$, equivalently its domain $\mathcal D_*(u)$, is a common compact partner support if, for every $i$:}
\begin{enumerate}
\item \revp{the logical partner algebra $\mathfrak L_{A_{B_i}}$ has a representative supported in a subdomain $d_i\subset\mathcal D_*(u)$;}
\item \revp{that representative is reconstructable from $R_i$, geometrically $d_i\subset {\rm EW}(R_i)$;}
\item \revp{the same domain $\mathcal D_*(u)$ is used for the whole family at the chosen epoch.}
\end{enumerate}
\end{definition}

\revp{Here $\mathfrak L_{A_{B_i}}$ denotes the logical operator algebra associated with the interior partner $A_{B_i}$ in the semiclassical code subspace. It should not be identified with the local operator algebra of a sharply defined spacetime region.
We suppress $u$ below when no confusion can arise. The support may evolve with $u$. It is therefore common only with respect to the chosen family of decoders at that epoch.}

\revp{For any achronal representative $X$ used below, $S_{\rm ren}(X)$ denotes its renormalized semiclassical matter entropy, evaluated with one fixed renormalization prescription. The same prescription is used for the spatial and null representatives of a given support domain.}

\begin{definition}[Bounded semiclassical null realization]
\label{def:bounded-null}
A bounded semiclassical null realization of $\mathcal D_*$ is an achronal null hypersurface $L_*\subset\Msc$ satisfying all of the following conditions:
\begin{enumerate}
\item $L_*$ is generated inward from the edge $\sigma_*$ along the null direction relevant to the support domain;
\item the closure $\overline{L_*}$ is compact in $\Msc$, every generator ends on a regular cut or caustic within $\Msc$, and no generator is required to reach a Planckian singularity or a boundary where semiclassical EFT has already failed;
\item \revp{$L_*$ and $B_*$ have the same domain of dependence within the cutoff spacetime,}
\begin{equation}
\Dsc(L_*)=\Dsc(B_*)=\mathcal D_*;
\label{eq:domain-equality}
\end{equation}
\item \revp{let $\sigma_0=\sigma_*$ and $\sigma_1$ denote the initial and terminal cuts of $L_*$. Choose half-Cauchy representatives $\Sigma_0$ and $\Sigma_1$ bounded by these cuts and differing by the complete null segment. With one fixed subtraction prescription, define the finite quantum-lightsheet entropy}
\begin{equation}
\revp{S_{\rm q}(L_*)\equiv S_{\rm out}(\Sigma_1)-S_{\rm out}(\Sigma_0).}
\label{eq:qbb-entropy}
\end{equation}
\revp{Because $L_*$ is a complete representative of the same support domain and the evolution is unitary in the semiclassical code subspace, this null entropy represents the same renormalized domain entropy as $B_*$,}
\begin{equation}
\revp{S_{\rm q}(L_*)=S_{\rm ren}(B_*).}
\label{eq:qbb-domain-match}
\end{equation}
\end{enumerate}
\revp{Condition 2 is what we mean by a bounded semiclassical closure.}
\end{definition}

Definition~\ref{def:bounded-null} supplies the precise meaning needed for the null comparison. It is neither a claim about every black-hole interior nor an assumption that the singularity is absent globally. It asks only whether the compact support proposed for the AMPS partners has at least one regular, bounded null representative \revp{with the same domain of dependence} in the same semiclassical region.
\revp{Equation~\eqref{eq:qbb-domain-match} uses the fact that \(B_*\) and \(L_*\) are complete achronal representatives of the same semiclassical domain of dependence. Under the assumed unitary evolution within that domain, passing from the spatial slice \(B_*\) to the null hypersurface \(L_*\) changes the representation of the same quantum state, but not the quantum information associated with the domain. This is the same domain-of-dependence logic used in entropy-based singularity arguments \cite{Bousso:2022cun}. The entropy assigned to the null representative is therefore identified with the renormalized entropy of the corresponding spatial support. As defined in Eq.~\eqref{eq:qbb-entropy}, the null entropy is the finite difference between the entropy on the final and initial cuts of the null hypersurface. This cut-to-cut quantity is the one constrained by the quantum Bousso bound~\cite{Bousso:2014sda,Strominger:2003br,Bousso:2015mna}.}

\section{From a QES edge to a quantum-lightsheet candidate}
\label{sec:qes}

The entropy comparison requires a null representative whose generalized expansion is nonpositive. \revp{Here the generalized expansion $\Theta_{\rm gen}$ is the null variation of generalized entropy per unit transverse area.} Quantum extremality provides the initial condition, while restricted quantum focusing controls the subsequent sign.

\begin{definition}[Quantum lightsheet]
\label{def:quantum-lightsheet}
Let $\sigma$ be a codimension-two surface and let $L$ be a null hypersurface generated from $\sigma$ along a chosen null direction. The relevant segment of $L$ is a quantum lightsheet if
\begin{equation}
\Theta_{\rm gen}\leq 0
\label{eq:quantum-lightsheet}
\end{equation}
everywhere on that segment.
\end{definition}

\begin{lemma}[A QES edge seeds a quantum lightsheet]
\label{lem:qes-lightsheet}
Let the edge $\sigma_*=\partial B_*$ of the proposed island support be a compact quantum extremal surface~\cite{Engelhardt:2014gca}, so that the generalized expansion in the relevant null direction satisfies
\begin{equation}
\Theta_{\rm gen}(\sigma_*)=0.
\label{eq:qes-expansion}
\end{equation}
Assume continuity of $\Theta_{\rm gen}$ along the generators and the restricted quantum focusing condition, namely that a nonpositive generalized expansion cannot cross to positive values along the corresponding null hypersurface~\cite{Shahbazi-Moghaddam:2022hbw}. Then the inward null hypersurface generated from $\sigma_*$ obeys $\Theta_{\rm gen}\leq0$ throughout the semiclassical segment and is a quantum-lightsheet candidate.
\end{lemma}

\begin{proof}
The initial cut has $\Theta_{\rm gen}=0$. If the expansion became positive at a later cut, continuity would imply a first crossing through zero with increasing generalized expansion. This is excluded by restricted quantum focusing. Hence $\Theta_{\rm gen}\leq0$ on the relevant segment.
\end{proof}

The lemma gives a lightsheet candidate; it does not prove bounded closure. \revp{Boundedness and equality of the domains of dependence are separate regularity requirements given in Definition~\ref{def:bounded-null}.} This separation is important: quantum focusing controls the sign of the expansion, whereas the global endpoint structure of the generators determines whether a bounded null representative exists.

\section{Entropy obstruction and the common-support dichotomy}
\label{sec:obstruction}

We now state all assumptions used in the entropy argument. %They are separated so that the theorem has no implicit geometric or information-theoretic input.

\begin{assumption}[Entropy loading by distinct late modes]
\label{ass:entropy-loading}
There is a family of operationally distinguishable late Hawking modes $\{B_i\}_{i=1}^N$ whose \revp{logical partner algebras $\{\mathfrak L_{A_{B_i}}\}_{i=1}^N$ are encoded in $\mathcal A(\mathcal D_*)$. In the finite-dimensional semiclassical code subspace, let $a_i$ denote the logical subsystem carrying the $i$th partner sector and let $a_{<i}=a_1\cdots a_{i-1}$. \revp{We use the standard conditional entropy\footnote{Writing the loading condition through conditional entropies avoids assuming that ordinary marginal entropies are additive.} $S(a_i\mid a_{<i})\equiv S(a_1\cdots a_i)-S(a_1\cdots a_{i-1})$.} We assume that genuinely new logical information is loaded by each distinguishable late mode in the sense that
\begin{equation}
S(a_i\mid a_{<i})\geq s_0>0
\label{eq:conditional-loading}
\end{equation}
for the family under consideration, up to subleading corrections controlled by code redundancy. By the chain rule,
\begin{equation}
S(a_1\cdots a_N)=\sum_{i=1}^N S(a_i\mid a_{<i})\geq Ns_0.
\label{eq:logical-chain}
\end{equation}
The renormalized entropy of the \revp{common} support is assumed to contain this logical contribution, so that
\begin{equation}
S_{\rm ren}(B_*)\geq S(a_1\cdots a_N)\geq Ns_0.
\label{eq:SB-lower}
\end{equation}}
\end{assumption}

\begin{remark}[What the lower bound counts]
\label{rem:counting}
The index $i$ labels distinct late Hawking modes, not different reconstructions of the same logical partner. Quantum error correction allows one logical operator to have redundant representatives in different physical regions~\cite{Almheiri:2014lwa,Dong:2016eik}. It does not identify operationally distinct logical modes while preserving their distinguishability under decoding. The assumption therefore counts independent logical late-mode sectors supported in $\mathcal D_*$, not the number of radiation regions from which one sector can be reconstructed~\cite{Hayden:2007cs,Harlow:2014yka}. 
\end{remark}

\begin{remark}[Semiclassical validity]
Assumption~\ref{ass:entropy-loading} is not an assertion of exact microscopic tensor factorization in quantum gravity. It is a code-subspace assumption appropriate to the same semiclassical regime in which AMPS treats separated late wave packets as distinguishable excitations. \revp{It does not assume exact tensor factorization of the full quantum-gravity Hilbert space. It assumes only a finite-dimensional code-subspace regime in which each operationally distinguishable partner contributes a positive conditional entropy of at least $s_0$ over the range of modes used in the theorem.}
\end{remark}

\begin{assumption}[Area control]
\label{ass:area-control}
\revp{At the chosen evaporation epoch $u$, the edge of the common support remains associated with the evaporating black hole and satisfies}
\begin{equation}
A(\sigma_*)=A(\partial B_*)\leq A_{\rm BH}(u).
\label{eq:area-control}
\end{equation}
\revp{Both $A_{\rm BH}(u)$ and the support itself may change with $u$.}
\end{assumption}

\begin{assumption}[\revp{Bounded null realization}]
\label{ass:bounded-null}
At the chosen epoch, the common support domain $\mathcal D_*$ admits a bounded semiclassical null realization $L_*$ in the sense of Definition~\ref{def:bounded-null}. Let $\sigma_0=\sigma_*$ and $\sigma_1$ be its initial and terminal cuts. The corresponding null segment is the quantum lightsheet supplied by Lemma~\ref{lem:qes-lightsheet}, and it obeys the quantum Bousso bound
\begin{equation}
S_{\rm q}(L_*)\leq
\frac{A(\sigma_0)-A(\sigma_1)}{4G}.
\label{eq:qbb-bound}
\end{equation}
This is the semiclassical entropy-bound input used in the theorem. It is the finite, fine-grained quantum Bousso bound obtained from the quantum focusing conjecture when the initial quantum expansion is nonpositive~\cite{Bousso:2015mna}. In our setup, this condition is natural because the initial cut is a QES, so $\Theta_{\rm gen}(\sigma_0)=0$. Restricted quantum focusing then implies that $\Theta_{\rm gen}$ does not become positive along the null segment considered here.

Related quantum entropy bounds are known in other semiclassical settings. A vacuum-subtracted quantum Bousso bound has been proved for free fields in the weak-backreaction regime without assuming the null energy condition~\cite{Bousso:2014sda}. Earlier work by Strominger and Thompson proposed a quantum covariant entropy bound~\cite{Strominger:2003br}, and explicit semiclassical examples have also been studied in JT gravity~\cite{Franken:2023ugu}.
\end{assumption}

\revp{The quantum bound itself constrains a finite entropy difference between null cuts. Definition~\ref{def:bounded-null} makes the necessary same-domain identification explicit: for a complete bounded representative, that cut-to-cut entropy is the renormalized entropy of the support domain. Thus the theorem does not replace a cut-to-cut quantum bound by an absolute entropy bound; it applies the quantum bound directly to a null representative of the full support domain.}

\revp{Following the entropy singularity theorem of Bousso and Shahbazi-Moghaddam~\cite{Bousso:2022cun}, we use ``hyperentropic'' in the same sense: a spatial region $B$ is hyperentropic when its renormalized entropy exceeds one quarter of the area of its boundary in Planck units,
\begin{equation}
S_{\rm ren}(B)>\frac{A(\partial B)}{4G}.
\label{eq:hyperentropic-definition}
\end{equation}
For the \revp{common} support, $\sigma_*=\partial B_*$, so this definition becomes Eq.~\eqref{eq:hyperentropic}.} %This is the same renormalized entropy--area criterion used in the entropy singularity theorem. What is new here is the mechanism that drives a proposed \rev{common partner support} across that threshold.

The motivation for Assumption~\ref{ass:bounded-null} is limited and specific. \revp{A common support proposed as one regular semiclassical interior should admit a controlled bounded null representative of the same domain before any Planckian breakdown.} This does not require that the full black-hole interior, or a maximally extended generator, avoid the singularity. \revp{If no such representative exists, the common support has already failed this regularity criterion; if one exists and the assumed entropy bound applies, the entropy comparison below follows. Thus the theorem is conditional in exactly the relevant sense.}

\begin{proposition}[Hyperentropic crossover]
\label{prop:hyperentropic}
Under Assumptions~\ref{ass:entropy-loading} and~\ref{ass:area-control}, if
\begin{equation}
Ns_0>\frac{A_{\rm BH}(u)}{4G},
\label{eq:crossover}
\end{equation}
then the \revp{common} support is hyperentropic in the sense of Eq.~\eqref{eq:hyperentropic-definition}:
\begin{equation}
\revp{S_{\rm ren}(B_*)}>\frac{A(\sigma_*)}{4G}.
\label{eq:hyperentropic}
\end{equation}
\end{proposition}

\begin{proof}
Assumption~\ref{ass:entropy-loading} gives $\revp{S_{\rm ren}(B_*)}\geq Ns_0$, while Assumption~\ref{ass:area-control} gives $A(\sigma_*)\leq A_{\rm BH}(u)$. Therefore,
\begin{equation}
\revp{S_{\rm ren}(B_*)}\geq Ns_0>
\frac{A_{\rm BH}(u)}{4G}
\geq\frac{A(\sigma_*)}{4G}.
\end{equation}
\end{proof}

\begin{theorem}[\revp{Common-support no-go}]
\label{thm:nogo}
\revp{Under Assumptions~\ref{ass:entropy-loading}--\ref{ass:bounded-null}, no regular common compact partner support satisfying the stated bounded-realizability and entropy-bound conditions can persist once Eq.~\eqref{eq:crossover} is reached.}
\end{theorem}

\begin{proof}
Proposition~\ref{prop:hyperentropic} gives
\begin{equation}
\revp{S_{\rm ren}(B_*)}>\frac{A(\sigma_*)}{4G}.
\label{eq:proof-hyper}
\end{equation}
Assumption~\ref{ass:bounded-null} gives the same-domain matching $S_{\rm q}(L_*)=S_{\rm ren}(B_*)$ and the quantum Bousso bound. Hence
\begin{equation}
\revp{S_{\rm ren}(B_*)}=S_{\rm q}(L_*)
\leq\frac{A(\sigma_*)-A(\sigma_1)}{4G}
\leq\frac{A(\sigma_*)}{4G},
\end{equation}
where $A(\sigma_1)\geq0$. This contradicts Eq.~\eqref{eq:proof-hyper}.
\end{proof}

\begin{corollary}[\revp{Common-support dichotomy}]
\label{cor:dichotomy}
Assume the first two assumptions. Once Eq.~\eqref{eq:crossover} holds, the common support is hyperentropic. At that stage, at least one of the following must occur:
\begin{enumerate}
\item the proposed common support no longer contains representatives of all partner sectors in the chosen family;
\item the support does not admit a bounded null realization with the same domain of dependence within $\Msc$;
\item the quantum Bousso bound does not apply to this null realization.
\end{enumerate}
If the quantum Bousso bound is assumed to remain valid and a bounded same-domain null realization is required for a regular semiclassical AMPS interior, then the common support cannot exist.
\end{corollary}

\begin{proof}
Assumptions 1 and 2 imply Eq.~\eqref{eq:hyperentropic},
\[
S_{\rm ren}(B_*)>\frac{A(\partial B_*)}{4G}.
\]
If the same support also admitted a bounded null realization with the same domain of dependence, Eqs.~\eqref{eq:qbb-domain-match} and~\eqref{eq:qbb-bound} would instead give
\[
S_{\rm ren}(B_*)\leq \frac{A(\partial B_*)}{4G}.
\]
The two inequalities are incompatible. Therefore the common support, the bounded same-domain null realization, and the validity of the quantum Bousso bound cannot all hold at the same time.
\end{proof}

Corollary~\ref{cor:dichotomy} also clarifies the role of entropy-based singularity arguments. We use ``hyperentropic'' in the same sense as Bousso and Shahbazi-Moghaddam~\cite{Bousso:2022cun}: the renormalized entropy of a spatial region exceeds one quarter of its boundary area in Planck units. In entropy-based singularity arguments, a hyperentropic region can obstruct regular semiclassical null development~\cite{Bousso:2022cun,Bousso:2022tdb}. In our setup, this possibility corresponds to point 2 above. The null hypersurface may fail to close regularly within $\Msc$, or its generators may leave the semiclassical regime. This does not invalidate the theorem. It means that the proposed common support does not admit the regular bounded null realization required in the argument. The proof does not rely on a pre-existing Planckian singularity; all entropy inequalities and domain-of-dependence comparisons are made within the semiclassical region.

\begin{figure}[t]
\centering
\begin{tikzpicture}[scale=0.95]
\draw[thick] (0,0) circle (1.15);
\node at (0,0.05) {$\rev{B_*}$};
\foreach \x/\y in {-0.45/0.45,0.35/0.55,-0.15/-0.35,0.55/-0.35} {
  \fill (\x,\y) circle (1.9pt);
}
\node at (0,-1.55) {partners $A_{B_i}$};
\node at (-3.00,0.55) {$B_1,B_2,\ldots,B_N$};
\draw[->] (-1.65,0.45) -- (-1.2,0.35);
\draw[->] (-1.65,0.25) -- (-0.55,0.45);
\draw[->] (-1.65,0.05) -- (0.5,-0.25);
\node at (3.5,0.55) {$\revp{S_{\rm ren}(B_*)}\gtrsim Ns_0$};
\node at (3.5,0.1) {$A(\sigma_*)\leq A_{\rm BH}(u)$};
\draw[->] (1.25,0.35) -- (1.9,0.35);
\draw[->] (1.25,-0.1) -- (1.9,-0.1);
\end{tikzpicture}
\caption{The entropy loading counts independent late Hawking modes and their logical partners, not multiple reconstructions of one partner. Once the \revp{common} support becomes hyperentropic, \revp{common support}, bounded semiclassical null realizability, and \revp{validity of the entropy bound} cannot all be maintained.}
\label{fig:fixed-support}
\end{figure}

\section{Discussion and conclusion}
\label{sec:discussion}

\subsection{Why the bounded realization is an AMPS regularity condition}

The bounded-null condition is not imposed on generic black-hole interiors. It applies only to the support domain used to represent the AMPS partners inside $\Msc$. The AMPS contradiction is formulated in regular near-horizon EFT, before an infalling observer reaches Planckian curvature. A common compact support proposed to represent the corresponding interior partners must therefore admit a controlled semiclassical description in that same regime. Requiring at least one AMPS-relevant bounded null realization provides a test of this regularity.

The existence of a singularity in the global black-hole spacetime does not affect this requirement. Definition~\ref{def:bounded-null} allows regular cutoff boundaries of $\Msc$ before semiclassical EFT breaks down. It requires only that the finite AMPS support domain admit a bounded null realization within this semiclassical region. If no such realization exists, the proposal falls into the second branch of Corollary~\ref{cor:dichotomy}. If it does exist, the entropy-bound contradiction follows after the hyperentropic crossover. Thus bounded null realizability is an explicit regularity assumption on the proposed common support, not a statement about the global completion of the black-hole interior.

\subsection{Region dependence and quantum error correction}

The result does not challenge the ordinary island prescription in Eq.~\eqref{eq:island-formula}. Standard islands remain region-dependent: different radiation systems may have different quantum extremal surfaces, Page times, and entanglement wedges.

The common-support condition is weaker than requiring every radiation subsystem to reconstruct the whole support. Each $R_i$ need reconstruct only its own partner sector. The obstruction therefore does not arise because any one radiation subsystem is asked to reconstruct too much of the interior. It arises because an increasing number of distinct partner sectors are placed in the same finite semiclassical support.

The entropy estimate also does not count the same logical partner several times because it has different reconstructions. Such redundancy is expected in holographic quantum error correction. The lower bound concerns distinct late Hawking modes $B_1,\ldots,B_N$ and their corresponding logical partner sectors. Redundant encoding can change where a logical sector is represented, but it does not identify distinct logical sectors with one another. This is the content of Assumption~\ref{ass:entropy-loading}.

The physical conclusion is therefore limited but sharp. The subsystem dependence of the AMPS reconstruction cannot, under the stated assumptions, be replaced by one finite common support that remains regular, bounded, and compatible with the quantum entropy bound.

\subsection{Time dependence}

The common support is defined at a fixed evaporation epoch $u$. Both $N(u)$ and $A_{\rm BH}(u)$ may change during evaporation, and the support itself may move. The theorem applies at any epoch at which Eq.~\eqref{eq:crossover} is satisfied. It does not require the same spatial region to remain fixed throughout the evaporation history.

\subsection{Solvable island models}

Solvable island models can avoid the obstruction in the expected way. Different radiation subsystems can enter island phases at different times, and there need not be a single epoch at which all AMPS-relevant partners are represented in one compact support~\cite{Geng:2020fxl,Geng:2020qvw,Geng:2021mic}. Recent analyses also emphasize that information encoding in island models can depend on the particular setup and need not define a single common-support interior map~\cite{Geng:2025rov,Geng:2025efs,Bao:2025plr}.

The theorem therefore shows that a consistent model must retain region-dependent reconstruction, give up a bounded common support, or leave the semiclassical regime.

\subsection{Limitations}

The result is semiclassical and conditional. The quantum Bousso bound in Assumption~\ref{ass:bounded-null} is the semiclassical entropy-bound input. The positive conditional-entropy condition in Assumption~\ref{ass:entropy-loading} is a code-subspace assumption rather than a theorem of exact quantum gravity. The cut-to-cut quantum Bousso bound constrains a finite entropy difference, while the same-domain identification in Eq.~\eqref{eq:qbb-domain-match} is imposed separately. A fully quantum treatment would require a regulator-independent algebraic formulation of this identification and of the corresponding generalized-entropy differences.

These limitations define the scope of the result. The paper rules out a regular semiclassical common compact partner support under the stated assumptions. It does not rule out ordinary region-dependent islands or a microscopic description of the interior beyond semiclassical gravity.

\vspace{2mm}

To summarize, we have shown that a common compact support containing an increasing number of distinguishable AMPS partner sectors eventually becomes hyperentropic. Once
\[
Ns_0>\frac{A_{\rm BH}}{4G},
\]
its renormalized entropy exceeds the area allowed by the support boundary. If the same support admits a bounded quantum lightsheet with the same domain of dependence, the quantum Bousso bound gives the opposite inequality. The common support, bounded null realizability, and validity of the quantum entropy bound therefore cannot all hold at the same time.

The main point is not that islands fail to resolve AMPS. Rather, the mode-by-mode, region-dependent reconstructions cannot in general be replaced by one compact semiclassical support for all late-mode partners. In this sense, the region dependence of island reconstruction is part of the way the semiclassical framework avoids the common-support obstruction.

\section*{Acknowledgments}

I thank Victor Franken for helpful comments, in particular for pointing out the relevance of the restricted quantum focusing condition to the lightsheet construction. I also thank Hao Geng for useful comments on the relation between the present obstruction and explicit solvable island models.

\bibliography{bib}
\bibliographystyle{JHEP}

\end{document}